\def\year{2017}\relax
\documentclass[letterpaper]{article}
\usepackage{cite}
\usepackage{aaai17}
\usepackage{times}
\usepackage{helvet}
\usepackage{courier}
\usepackage{graphicx}
\usepackage{amsmath}
\usepackage{amsthm}
\usepackage{color}
\usepackage{amsfonts}
\usepackage{mathrsfs}

\long\def\comment#1{}

\begin{document}

%
\title{Mechanism Design in Social Networks}

\author{Bin Li$^a$, Dong Hao$^a$, Dengji Zhao$^b$ and Tao Zhou$^a$\\
$^a$Big Data Research Center, University of Electronic Science and Technology of China, Chengdu, China\\
\{libin@std.uestc., haodong@uestc.,  zhutou@ustc.\}edu.cn\\
$^b$School of Information Science and Technology, ShanghaiTech University, Shanghai, China\\
zhaodj@shanghaitech.edu.cn\\
}


\maketitle


\newtheorem{theorem}{Theorem}
\theoremstyle{theorem}

\newtheorem{lemma}{Lemma}
\theoremstyle{theorem}
\newtheorem{theo}{Theorem}

\newtheorem{corollary}{Corollary}
\theoremstyle{theorem}

\theoremstyle{plain}

\newtheorem{prop}{Proposition}

\begin{abstract}
This paper studies an auction design problem for a seller to sell a commodity in a social network, where each individual (the seller or a buyer) can only communicate with her neighbors. The challenge to the seller is to design a mechanism to incentivize the buyers, who are aware of the auction, to further propagate the information to their neighbors so that more buyers will participate in the auction and hence, the seller will be able to make a higher revenue. We propose a novel auction mechanism, called information diffusion mechanism (IDM), which incentivizes the buyers to not only truthfully report their valuations on the commodity to the seller, but also further propagate the auction information to all their neighbors. In comparison, the direct extension of the well-known Vickrey-Clarke-Groves (VCG) mechanism in social networks can also incentivize the information diffusion, but it will decrease the seller's revenue or even lead to a deficit sometimes. The formalization of the problem has not yet been addressed in the literature of mechanism design and our solution is very significant in the presence of large-scale online social networks.
\end{abstract}

\section{Introduction}
Mechanism design is a representative interface integrating economics and artificial intelligence \cite{nisan2007algorithmic}. It utilizes game theoretic tools to model interactions between agents, and takes a systematic investigation of the design of institutions and how these institutions affect the outcomes of agents' interactions. Auction has been a common paradigm \cite{krishna2009auction} and a successful application \cite{Edelman2007Internet,bajari2003winner} in mechanism design. Although in almost every human-activated system, the social network matters \cite{jackson2008social,borgatti2009network}, the standard auction theory only considers buyers who are called together by the seller, while the effects of social links connecting these buyers to other potential buyers who are not known by the seller have not yet been explored. Without an effective auction mechanism that facilitates information propagation in the social network, a potential buyer with a high valuation may not be aware of the auction since she is not directly connected to the seller and her neighbors do not wish to tell her the auction information. Therefore, in order to simultaneously improve the seller's revenue and the social welfare, which are two primary and generally conflicting objectives in any auction design \cite{myerson1981optimal,Krishna1998Efficient}, it is of great importance to design auction mechanisms in social networks, which not only satisfy the classical criteria such as incentive compatibility and individual rationality, but also incentivize individuals to spread the auction information in the social network.

The main difficulty for incentivizing the individuals to spread the auction information lies in the general conflict that the seller wishes to attract more people to participate in the auction in order to increase her revenue, but participants have no incentive to bring more competitors into the auction. This essentially reflects the conflict between the system's optimality and the individuals' self-interests. Under all classic auction protocols including the Vickrey-Clarke-Groves (VCG) mechanism \cite{vickrey1961counterspeculation,clarke1971multipart,groves1973incentives}, the seller can only convene a small amount of people into the auction while the other potential buyers in the social network are all excluded from the auction, even though they may have higher valuations. As a result, the seller's revenue can only be locally optimized, or, the efficiency (social welfare) can only be locally maximized.

The goal of this paper is to design mechanisms that can improve both the efficiency and the seller's revenue, which is not achievable with existing mechanisms. To this end, we propose a novel mechanism that incentivizes all participants to further propagate the auction information, which will lead to a more efficient allocation and more revenue to the seller. The intuition behind this mechanism is that information diffusion is rewarded if it leads to a more efficient allocation.
Our contributions advance the state of the art in the following ways:
\begin{itemize}
\item We propose the very first model for selling one item in a social network, where each participant in the network can only communicate with her neighbors.
\item In the network setting, we first show that the well-known VCG mechanism can be extended to incentivize information diffusion and therefore all potential buyers in the network will join the auction to achieve the optimal social welfare. However, it does not increase the revenue of the seller and even leads to a deficit sometimes, which will disincentivize the seller to apply such a mechanism.  
\item  Therefore, we propose another novel mechanism called information diffusion mechanism (IDM) to solve the deficit issue of the VCG. IDM not only incentivizes information diffusion as the VCG does, but also brings more revenue to the seller. In particular, the revenue generated by IDM is always above the revenue created by the VCG and it is non-negative. 
\end{itemize}

The remainder of the paper is organized as follows. We first define the bidder's actions and the seller's revenue for a single item auction in a social network, extend the classical incentive compatibility to the social network setting, and then define some general concepts of an auction mechanism in such a framework. We then generalize the well-known VCG mechanism into the social network setting and analyze its performance. Following the weak performance of the VCG, we propose our information diffusion mechanism (IDM) and prove that it has a remarkable performance compared to the VCG. Finally, we summarize this work and discuss future work.

\section{The Model}
\theoremstyle{definition}
\newtheorem{defn}{Definition}

We consider a social network which consists of $n$ agents denoted by a set $N = \{1, \cdots, n\}$.  For each agent $i\in N$, she has a set of neighbors $r_i\subseteq N\setminus \{i\}$ with whom $i$ can directly exchange information (assume that $r_i\neq \emptyset$). One agent in this network, called the seller $s$, has an item to sell. Each agent $i\in N$ has a valuation $v_i \geq 0$ on the item. For the seller $s$, $v_s$ represents her reserve price, which is assumed to be zero in this paper. Initially, only the seller's neighbors $r_s$ know that $s$ is selling the item.

Given the above setting, the goal of the seller is to get more agents from the network to participate in the auction and thus to maximize her revenue. Since the seller cannot directly inform all agents other than $r_s$, the challenge to her is to build a mechanism to incentivize agents who have received the auction information to further propagate the information to their neighbors, which has not yet been tackled in the literature.

In this paper, the auction mechanism asks each agent to both report her valuation on the item and to spread the auction information to her neighbors. Accordingly, an agent $i$'s action in the mechanism is defined as $a_i'=(v_i', r_i')$, where $v_i'$ represents the valuation she reports and $r_i'\subseteq r_i$ represents the neighbors to whom she tells the auction information. Let $a_i=(v_i,r_i)$ denote the \emph{truthful action} that $i$ truthfully declares her valuation and informs the auction information to all her neighbors. We assume that once an agent received the auction information, 
she will be able to participate in the auction by reporting her valuation and further spreading the auction information.  

Let $A_i = V_i\times \mathcal{P}(r_i) \cup \{null\}$ be $i$'s action space, where $V_i$ is $i$'s valuation space, $\mathcal{P}(r_i)$ is the power set of $r_i$ and $null$ is a dummy action for $i$ when $i$ has not received the auction information or $i$ does not want to participate in the auction. Let $N_{-s} = N\setminus \{s\}$, ${\bf a'} = (a_i'\in A_i)_{i\in N_{-s}}$ be an action profile of all buyers, ${\bf a'}_{-i}$ be the action profile of all buyers except $i$ and ${\bf a'} = (a_i', {\bf a'}_{-i})$.

\begin{defn}
We say an action profile ${\bf a'}$ is \emph{feasible} if $a_i' = null$ for all $i\in N_{-s}$ who cannot receive the auction information following the actions of the others.
\end{defn}

Let $A$ be the set of all feasible action profiles of all buyers $N_{-s}$, and $A_{-i}^{a'_i} = \{{\bf a'}_{-i}| (a'_i, {\bf a'}_{-i}) \in A\}$ be the set of all feasible action profiles of all buyers except $i$ given that $i$'s action is $a'_i$. Given any non-feasible action profile ${\bf a'}$, we define a function $f$ that transforms ${\bf a'}$ to a feasible action profile $f({\bf a'}) = {\bf a''} \in A$, where
\begin{equation*}
 a_i'' =
 \begin{cases}
    a_i' & \text{if $i$ receives the information following ${\bf a'}_{-i}$},\\
    null           & \text{otherwise.}\\
 \end{cases}
\end{equation*}

Given the buyers' action profile, we formally define the auction mechanism and the related concepts in the rest of this section.

\begin{defn}
An \emph{auction mechanism} $\mathcal{M}$ in the social network is defined by an allocation policy $\pi = \{\pi_i\}_{i\in N_{-s}}$ and a payment policy $p=\{p_i\}_{i\in N_{-s}}$, where $\pi_i:A\rightarrow \{0,1\}$ and $p_i:A\rightarrow \mathbb{R}$ are the allocation and payment functions for $i$ respectively.
\end{defn}

Given the buyers' feasible action profile ${\bf a}^\prime$, $\pi_i({\bf a}^\prime) = 1$ means that $i$ receives the item, while $\pi_i({\bf a}^\prime) = 0$ means that $i$ does not receive the item. $p_i({\bf a}^\prime) \geq 0$ indicates that $i$ pays the auctioneer $p_i({\bf a}^\prime)$, and $i$ receives $|p_i({\bf a}^\prime)|$ from the auctioneer if $p_i({\bf a}^\prime) < 0$. We say an allocation policy $\pi$ is \emph{feasible} if for all ${\bf a}' \in A$, $\sum_{i\in N_{-s}} \pi_i({\bf a}') \leq 1$, and $\pi_i({\bf a}') = 1$ implies $a'_i \neq null$. That is, a feasible allocation can only allocate the item to at most one buyer who has participated in the auction. In what follows, we only consider feasible allocations.

\begin{defn}
An allocation $\pi$ is \emph{efficient} if for all ${\bf a}' \in A$,
$$\pi \in {\arg\max}_{\pi^\prime \in \Pi} \sum_{i\in N_{-s}, a'_i \ne null} \pi^\prime_i({\bf a}')v_i^\prime$$
where $\Pi$ is the set of all feasible allocations.
\end{defn}

Given a buyer $i$ of truthful action $a_i = (v_i, r_i)$, a feasible action profile ${\bf a}^\prime$ and a mechanism $\mathcal{M} = (\pi, p)$, the \emph{utility} of $i$ under the allocation $\pi({\bf a}^\prime)$ and the payment $p({\bf a}^\prime)$ is quasilinear and is defined as:
\begin{equation*}
u_i(a_i, {\bf a}^\prime, (\pi,p)) = \pi_i({\bf a}^\prime)v_i  - p_i({\bf a}^\prime).
\end{equation*}
We say a mechanism is individually rational if for each buyer, her utility is non-negative when she truthfully reports her valuation, no matter to whom she tells the auction information and what the others do.

\begin{defn}
A mechanism $(\pi,p)$ is \emph{individually rational} (IR) if $u_i(a_i, ((v_i, r_i^\prime),{\bf a}_{-i}^\prime), (\pi,p)) \geq 0$ for all $i\in N_{-s}$, all $r_i^\prime \in \mathcal{P}(r_i)$, and all ${\bf a}_{-i}^\prime \in A_{-i}^{(v_i, r_i^\prime)}$.
\end{defn}

Another concept is called incentive compatibility or truthfulness. In a traditional auction setting where a seller sells a commodity to multiple buyers, incentive compatibility means that for each buyer, reporting her valuation on the commodity truthfully is a dominant strategy. However, in our model, auction information diffusion is also considered. Here, incentive compatibility means that for each buyer, reporting her true valuation and diffusing the auction information to all her neighbors is a dominant strategy.

\begin{defn}
 A mechanism $(\pi, p)$ is \emph{incentive compatible} (IC) if
 $u_i(a_i, (a_i,{\bf a}_{-i}^\prime), (\pi,p)) \geq u_i(a_i, (a_i^\prime, {\bf a}_{-i}^{\prime\prime}), (\pi,p))$ for all $i\in N_{-s}$, all $a_i^\prime \in A_i$, and all ${\bf a}_{-i}^\prime \in A_{-i}^{a_i}$, where $(a_i^\prime, {\bf a}_{-i}^{\prime\prime}) = f(a_i^\prime, {\bf a}_{-i}^\prime)$.
\end{defn}

Note that in this definition, ${\bf a}_{-i}^\prime$ is changed to ${\bf a}_{-i}''$ when $i$'s action is changed from $a_i$ to $a_i^\prime$. This is because if an action $a_i^\prime$ does not spread the auction information to all $i$'s neighbors, then some buyers, who receive the information under $a_i$, may not be able to receive the information under $a_i^\prime$. That is, although $(a_i, {\bf a}_{-i}^\prime)$ is a feasible action profile, it does not guarantee that $(a_i^\prime, {\bf a}_{-i}^\prime)$ is also feasible. Therefore, in the IC definition, we apply the function $f$ to get the corresponding feasible action profile of the other buyers when $i$'s action is changed.

Given a feasible action profile ${\bf a}^\prime$ and a mechanism $\mathcal{M} = (\pi, p)$, the seller's \emph{revenue} generated by $\mathcal{M}$ is defined by the sum of all buyers' payments, denoted by $Rev^{\mathcal{M}}({\bf a}^\prime) = \sum_{i\in N_{-s}} p_i({\bf a}^\prime)$.

\begin{defn}
A mechanism $\mathcal{M}=(\pi, p)$ is \emph{weakly budget balanced} if for all ${\bf a}' \in A$,  $Rev^{\mathcal{M}}({\bf a}') \geq 0$.
\end{defn}

In the rest of this paper, we design mechanisms that satisfy the above concepts. Before doing so, we will introduce an important concept related to auction information diffusion, called diffusion critical node.

\begin{defn}
\label{diffusion critical node}
Given the buyers' feasible action profile ${\bf a}^\prime$, for any two buyers $i\neq j\in N_{-s}$ such that $a'_i ,a'_j \ne null$, we say $i$ is $j$'s \emph{diffusion critical node} if all the information diffusion paths started from the seller $s$ to $j$ have to pass $i$.
\end{defn}

Intuitively, if $i$ is $j$'s diffusion critical node under the action profile ${\bf a}^\prime$, $j$ will not be able to receive the auction information if $i$ has not propagated the information to her neighbors. Note that, a buyer $j$ may have no diffusion critical node (e.g. the seller's neighbors) or more than one diffusion critical nodes. Diffusion critical nodes will play a vital role in the mechanisms proposed in the paper.

\section{VCG in Social Networks}
Given the social network and the extended definition of incentive compatibility, we will first check whether traditional incentive compatible mechanisms can be extended to incentivize information diffusion. In this section, we implement the well-known Vickrey-Clarke-Groves (VCG) mechanism under the social network setting. We prove that although VCG mechanism can incentivize information diffusion, it will decrease the seller's revenue and even lead to a deficit sometimes.

The VCG mechanism applies an efficient allocation and charges each participant the social welfare decrease of the others caused by her participation \cite{vickrey1961counterspeculation,clarke1971multipart,groves1973incentives}.
In a classical setting where a seller sells one item to a set of buyers who are known to the seller in advance, the VCG mechanism allocates the item to the buyer who has the highest bid (the winner) and charges the buyer the second highest bid among all the buyers (also known as the second price auction). The other buyers who do not receive the item will not pay or get anything. However, in a social network setting, if a buyer does not get the item, but she is the winner's diffusion critical node, she will receive a reward. This is because the winner will not be able to participate in the auction if her diffusion critical buyer(s) does not spread the information.


Formally, given the buyers' action profile ${\bf a}' \in A$, for each buyer $i\in N_{-s}$ such that $a'_i \ne null$, let
\begin{equation*}
s_i = \{j\in N_{-s}| \text{$i$ is $j$'s diffusion critical node}\}
\end{equation*}
denote the set of all buyers who share $i$ as their diffusion critical node, and let
$d_i=s_i \cup \{i\}$.
If $i$ does not participate in the auction (i.e., $i$ chooses a dummy action $null$), then all buyers from $s_i$ also cannot participate in the auction as they will not be able to receive the auction information without $i$. Let $-d_i = N_{-s} \setminus d_i$ and ${\bf a}'_{-d_i}$ be the action profile of $-d_i$ in ${\bf a}'$. It is evident that ${\bf a}'_{-d_i}$ cannot be changed by $i$'s action.
Now the VCG mechanism allocates the item to the buyer with the highest valuation report (i.e., an efficient allocation) and the payment for each buyer $i$ is defined as:
\begin{equation}
\label{eq_vcg_pay}
p_i^{vcg}({\bf a}') = W({\bf a}'_{-d_i}) - (W({\bf a}')- \pi_i^*({\bf a}')v_i')
\end{equation}
where $W({\bf a}'_{-d_i}) = \sum\nolimits_{j \in -d_i, a'_j \neq null} \pi_j^*({\bf a}'_{-d_i})v_j'$, $W({\bf a}') = \sum\nolimits_{j \in N_{-s}, a'_j \neq null} \pi_j^*({\bf a}')v_j'$, and $\pi^*$ is an efficient allocation.

In the VCG payment definition, the first term $W({\bf a}'_{-d_i})$ is the social welfare of the efficient allocation without $i$'s participation, and the second term $W({\bf a}')- \pi_i^*({\bf a}')v_i'$ is the social welfare of the efficient allocation minus $i$'s valuation on the allocation (i.e. the social welfare of all buyers except $i$ on the efficient allocation). Therefore, the VCG payment for $i$ is the social welfare decrease of the other buyers caused by $i$'s participation (it can be negative).

In what follows, we prove that the VCG mechanism is individually rational and incentive compatible. Before doing this, we show one property of $d_i$ with respect to $r_i$ in Lemma~\ref{membership_di}.

\begin{lemma}\label{membership_di}
Given the buyers' action profile ${\bf a'}$, for each buyer $i$ and for any $r_i^\prime, r_i^{\prime\prime} \in \mathcal{P}(r_i)$, if $r_i^\prime \subseteq r_i^{\prime\prime}$, then the corresponding $d_i$, denoted by $d_i^\prime$ and $d_i^{\prime\prime}$, satisfy $d_i^\prime \subseteq d_i^{\prime\prime}$.
\end{lemma}
\begin{proof}
For all $j\neq i\in d_i^\prime$, we have that all the diffusion paths to reach $j$ have to pass $i$ when $i$ propagates the information to $r_i^\prime$. When $i$ changes the diffusion set from $r_i^\prime$ to $r_i^{\prime\prime}$, the old set of paths to reach $j$ will not be affected and new paths to reach $j$ might be created, but all the new paths have to pass $i$ (as the other buyers' actions have not been changed). Therefore, we have $d_i^\prime \subseteq d_i^{\prime\prime}$.
\end{proof}

%

\begin{theorem}
The VCG mechanism is individually rational and incentive compatible.
\end{theorem}
\begin{proof}
Given all buyers' true action profile ${\bf a}$ and their actual action profile ${\bf a}^\prime$, buyer $i$'s utility under the VCG mechanism $(\pi^*, p^{vcg})$ is $u_i(a_i,{\bf a'}, (\pi^*, p^{vcg}))= \pi_i^*({\bf a'})v_i - p_i^{vcg}({\bf a'}) = W({\bf a'}) + \pi_i^*({\bf a'})(v_i-v_i^\prime) - W({\bf a'}_{-d_i})$.

If buyer $i$ reports her valuation truthfully (i.e. $v_i^\prime = v_i$), we get $u_i(a_i,{\bf a'}, (\pi^*, p^{vcg})) = W({\bf a'}) - W({\bf a'}_{-d_i})$. It is evident that $W({\bf a'}) \geq W({\bf a'}_{-d_i})$ because the allocation $\pi^*$ is efficient. Therefore, $i$'s utility is non-negative when $i$ reports her valuation truthfully (no matter who she spreads the information to). That is, the VCG mechanism is individually rational.



Now we prove that the VCG mechanism is incentive compatible. We will prove this by showing that:
\begin{enumerate}
\item For each buyer $i$, fixing $r_i^\prime$, $i$'s utility is maximized by reporting $v_i$ truthfully.
\item For each buyer $i$, fixing $v_i^\prime$, $i$'s utility is maximized by diffusing the auction information to all her neighbors $r_i$.
\end{enumerate}

Given that $r_i^\prime$ is fixed (i.e. information diffusion is fixed), the first step is the traditional IC property of the VCG, which is proved by showing that (1) $i$'s payment does not depend on $i$'s reported valuation $v_i^\prime$ and (2) $i$'s utility is maximized by reporting $v_i$ truthfully. Since the VCG is known to be incentive compatible in this setting, we will omit the proof here. 

For the second step, since $-d_i$ is not affected by $i$'s action $r_i'$, $W({\bf a'}_{-d_i})$ does not depend on $r_i'$. By Lemma~\ref{membership_di}, we get that diffusing to more neighbors will get more buyers to participate in the auction, which will potentially increase $W({\bf a'})$, the highest valuation among the participated buyers. Therefore, for each buyer $i$, when $v_i^\prime$ is fixed, diffusing the auction information to all her neighbors $r_i$ maximizes her utility $W({\bf a'}) + \pi_i^*({\bf a'})(v_i-v_i^\prime) - W({\bf a'}_{-d_i})$.

By proving the above two properties, we conclude that for each buyer, reporting her valuation truthfully and diffusing the information to all her neighbors is a dominant strategy, i.e., the VCG is IC.
\end{proof}

\begin{prop}
\label{prop_vcg}
The VCG mechanism is not weakly budget balanced.
\end{prop}
\begin{proof}
We prove this by example. Figure~{\ref{line_graph}} shows a network where the seller $s$ and $l$ buyers form a line and the seller is located at one of the end points. Here, all buyers, except the one located at the other end of the line, have a valuation zero, and the only non-zero valuation is $1$. Applying the VCG on this setting, the buyer with valuation $1$ receives the item and her payment is zero. All the other buyers do not receive the item, but each of them will receive $1$ according to the payment policy, because without any of them, the buyer with valuation $1$ will not be able to participate in the auction. Thus, the revenue of the seller is $-(l-1)$.
%
\end{proof}

\begin{figure}[h]
  \centering
  \includegraphics[width=2.3in]{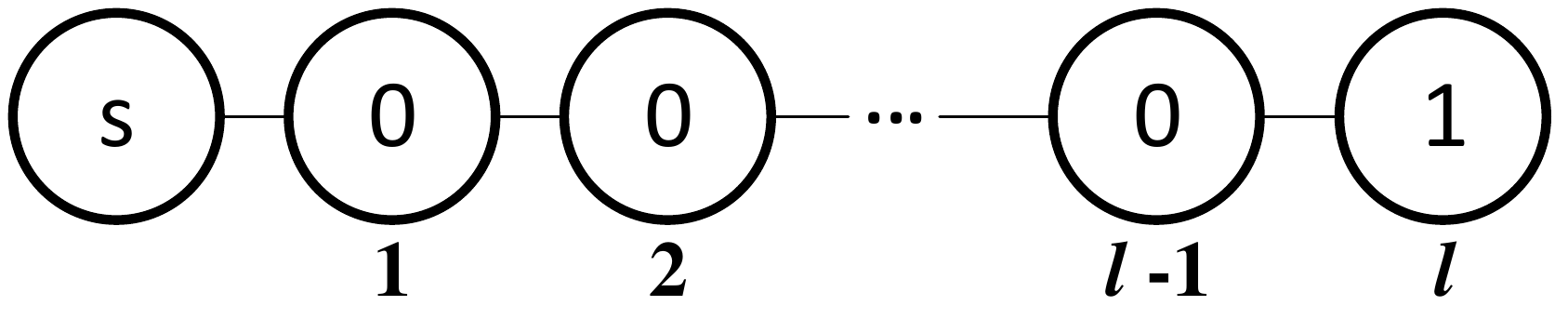}
  \caption{A line-structured social network for the proof of Proposition~\ref{prop_vcg}.}\label{line_graph}
\end{figure}

By contrast, for the same setting given in Figure~\ref{line_graph}, if the seller applies the second price auction (without diffusion incentive), then only one buyer near the seller participates in the auction, and the seller's revenue is zero. Therefore, the seller would rather prefer the second price auction to the VCG mechanism with information diffusion with respect to her revenue. The reason for the seller to incentivize buyers to spread the auction information is to obtain a higher revenue, but the VCG mechanism cannot achieve this. Hence, in the next section, we propose a novel solution which will not only incentivize buyers to spread the information, but also increase the seller's revenue.

\section{Information Diffusion Mechanism}
In this section, to conquer the revenue issue of the VCG mechanism in social networks, we design a mechanism that is incentive compatible and also increases the seller's revenue (weakly budget balanced). Based on the notion of diffusion critical node, we first give another definition of \emph{diffusion critical sequence}, which is a key element of our mechanism.

\begin{defn}\label{Diffusion_Critical_Sequence}
Given the buyers' action profile ${\bf a}^\prime$, for each $i\in N_{-s}$, we define $C_i=\{x_1,x_2,\cdots,x_k,i\}$ the \emph{diffusion critical sequence} of $i$, which is an ordered set of all diffusion critical nodes of $i$ and $i$ itself and the order is determined by the relation $d_{x_1} \supset d_{x_2} \supset ,\cdots, \supset d_{x_k}\supset d_{i} $.
\end{defn}

The diffusion critical sequence captures how the auction information from the seller is propagated to buyer $i$ through its diffusion critical nodes. It is essentially a partial route of all the information diffusion paths from the seller to buyer $i$, which clearly shows to what extend each node in the sequence affects the information diffusion process.
A node in the sequence can only receive the auction information if all nodes ordered above the node have received the information and propagated to the next node in the sequence.
It is not difficult to verify that, after the diffusion process, there exists one and only one diffusion critical sequence for each buyer participated in the auction. 

%


\begin{figure}[h]
  \centering
  \includegraphics[width=3.1in]{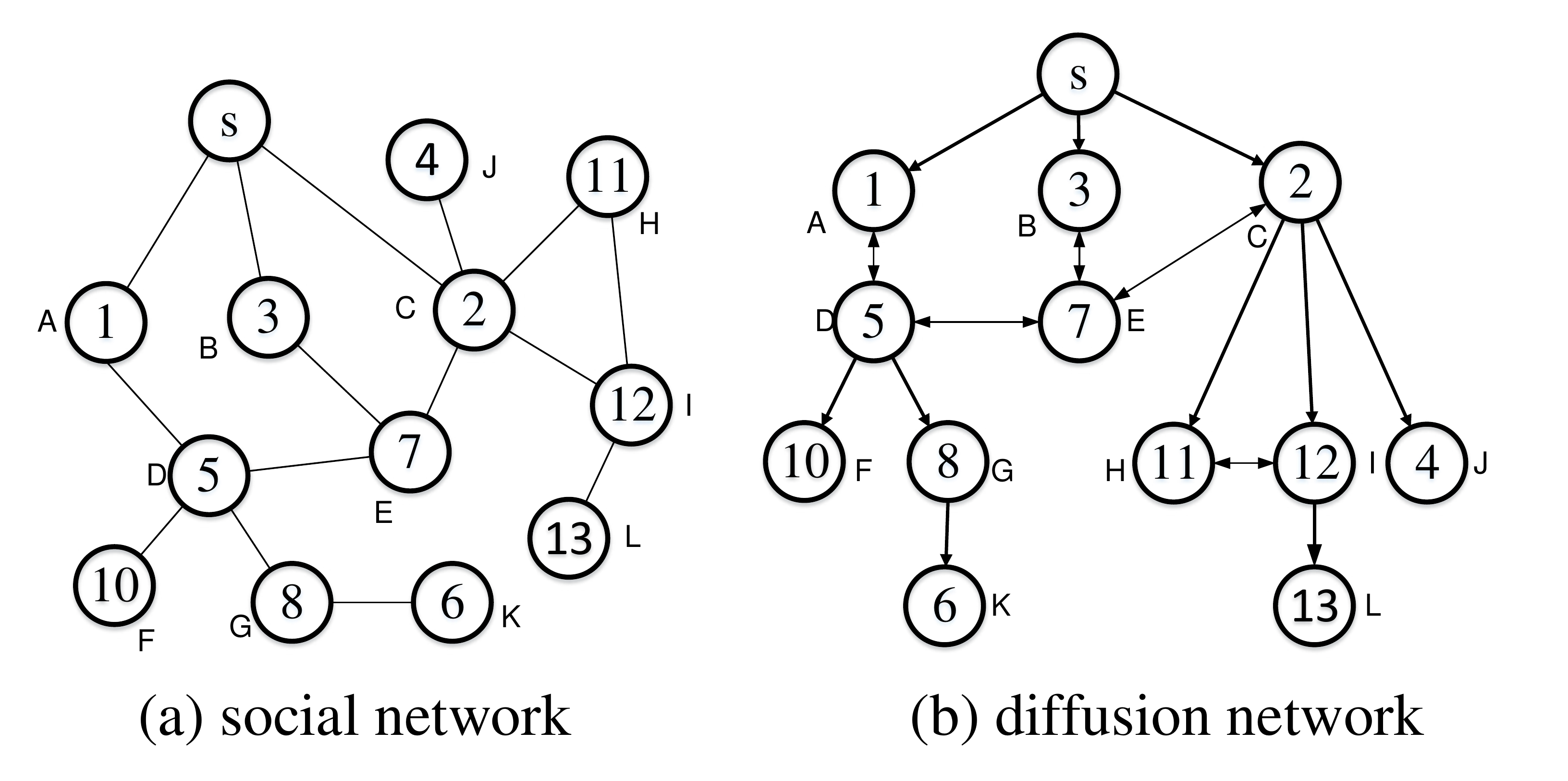}
  \caption{A social network example and the corresponding information diffusion flow.\label{information_graph}}
\end{figure}


To further demonstrate the property of a buyer's diffusion critical sequence, we study an example given in Figure~\ref{information_graph}. Figure 2(a) shows a simple social network, where each circle denotes an agent and the value (except $s$) inside each circle is the agent's valuation. The agent with value $s$ denotes the seller and all the potential buyers are labelled by $A, B, \cdots, L$. The links between circles represent the neighborhood relationship. Figure 2(b) shows the corresponding directed graph of the information flow, where the directed links indicate the information flow directions. If two buyers are linked by an one-way arrow edge, it means that the auction information can only flow in the direction of the arrow. For example, the information can flow from buyer $D$ to buyer $G$, but not the other way around. This is because $G$ cannot receive the information from its other neighbors other than $D$. Here the diffusion critical sequence of buyer $G$ is $C_G=\{D,G\}$. Note that buyers $A,B,C,E$ are not $G$'s diffusion critical nodes, although they are in the diffusion paths from the seller to $G$, because without any of them $G$ can still receive the information from the other paths. It is also easy to check that $d_D=\{D,F,G,K\} \supset d_G=\{G,K\}$.

Now we propose our mechanism based on the diffusion critical sequence.

\begin{defn}[\textbf{Information Diffusion Mechanism} (IDM)]
Given the buyers' action profile ${\bf a}'$, for any $B\subseteq N_{-s}$, let $v_B^*= \max_{i\in B, a_i^\prime \neq null} v_i^\prime$ denote the highest valuation report among buyers $B$, let $m$ denote the buyer with the highest valuation report in ${\bf a}'$ (with random tie-breaking), i.e., $v'_m = v_{N_{-s}}^*$. To simplify the representation, let $C_m = \{1,2,\cdots, i-1, i,i+1,\cdots,m-1,m\}$ be $m$'s diffusion critical sequence.

The allocation policy of the \emph{information diffusion mechanism} is defined as:
\begin{equation}\label{IDM_allocation_rule}
\pi_i^{idm} ({\bf a'}) =
\begin{cases}
 1 & \text{if $i \in C_m \setminus\{m\}$ and $v'_{i}=v_{{{-d}_{{i+1}}}}^{*}$,}\\
 1 & \text{if $i = m$,}\\
 0 & \text{otherwise.} \\
 \end{cases}
 \end{equation}
If there exist multiple buyers $i$ with $\pi_i^{idm}({\bf a'})=1$, allocate the item to the buyer with minimum index $i$ in $C_m$. Assume that under the allocation policy, buyer $w\in C_m$ wins the item, the payment policy is defined as:
\begin{equation}\label{IDM_payment_rule}
p_i^{idm} ({\bf a'}) =
\begin{cases}
 v_{{{-d}_{i}}}^{*}  - v_{{{-d}_{i+1}}}^{*} & \text{if $i \in C_w\setminus\{w\}$,}\\
 v_{{{-d}_{i}}}^{*} & \text{if $i=w$,} \\
 0 & \text{otherwise.} 
 \end{cases}
 \end{equation}

\end{defn}

Intuitively, IDM allocates the item to the first buyer $i$ in the diffusion critical sequence $C_m$ whose bid is the highest bid when diffusion critical node $i+1$ does not participate in the auction. 
The winner $w$ pays the highest bid without her participation. Each buyer in $C_w\setminus \{w\}$ (winner's diffusion critical node) is rewarded the payment increase (not social welfare increase) due to her diffusion action. In particular, if $i\in C_w\setminus \{w\}$ keeps the item by herself, she will pay something say $x$, while if she gives it to $i+1$ and $i+1$ keeps it, $i+1$ will pay a different amount say $y$, the difference $y-x$ is rewarded to $i$.

Before analyzing the properties of IDM, we show a running example of IDM by using Figure 2(b). In this setting, assume that all buyers truthfully report their valuations and diffuse the information to all their neighbors. It is clear that $L$ is the buyer with the highest valuation and the diffusion critical sequence of $L$ is $C_L=\{C,I,L\}$. According to the allocation policy, buyer $I$ wins the auction because $v_I = 12$ is the highest valuation when $L$ does not participate in the auction. Finally, according to the payment policy, buyer $I$ should pay ${v_{{{-d}_I}}^{*}} = 11$, and buyer $C$ is rewarded $1 = 11-10$.

In what follows, we study the properties of IDM. Following the definition of the mechanism, for a given action profile, we classify all buyers into four different status:
\begin{enumerate}
\item \emph{The winner}: $w$. 
\item \emph{On-path buyers}: all $w$'s diffusion critical nodes, i.e., $C_w\setminus\{w\}$.
\item \emph{Unlucky buyers}: if $w \ne m$, then all $d_{w+1}$ are unlucky buyers, otherwise, all $d_m \setminus \{m\}$ are unlucky buyers.
\item \emph{Normal buyers}: all buyers who are not classified in any of the other three status. 
\end{enumerate}
According to the payment policy of IDM, once the diffusion network has been established and $C_w$ has been identified, only the buyers in $C_w$ (i.e., the winner and the on-path buyers) are involved with non-zero money transfer. Specifically, for an unlucky buyer $i$, when $w \ne m$, we have $i \in d_{w+1}$ and action $a_i^\prime$ cannot change $v^*_{-d_{w+1}}$. Therefore, $v'_w = v^*_{-d_{w+1}}$ still holds no matter what $a_i^\prime$ is, i.e., buyer $w$ is still the winner and buyer $i$ is still an unlucky buyer for any $a_i^\prime$. Similarly we can conclude the same when $w=m$. In other words, the utility of any unlucky buyer is always zero no matter what action she takes.

Given the above classifications of buyers, we will prove that IDM is incentive compatible and individually rational. To do so, we first prove the following two lemmas.
Lemma \ref{not_depend_valuation} shows that under IDM a buyer's payment is independent of her valuation report.

\begin{lemma}\label{not_depend_valuation}
For all ${\bf a}^\prime \in A$, all $i\in N_{-s}$, $p_i^{idm}({\bf a'})$ is independent of $v_i^\prime$.
\end{lemma}
\begin{proof} Firstly, if $i$ is an unlucky buyer or a normal buyer (i.e., $i \in N_{-s}\setminus\{C_w\}$), her utility is zero which is independent of $v_i^\prime$. Secondly, if $i$ is an on-path buyer (i.e., $i \in C_w \setminus \{w\}$), her payment is $v_{{{-d}_{i}}}^{*}  - v_{{{-d}_{i+1}}}^{*}$. Since $i \notin {{{-d}_{i}}}$, $v_{{{-d}_{i}}}^{*}$ does not depend on $v_i'$. Since she is not the winner, we have $v_{{{-d}_{i+1}}}^{*}\ne v'_i$. Therefore, the payment of an on-path buyer is independent of $v_i^\prime$. Finally, if $i$ is the winner, her payment is $v_{{{-d}_{i}}}^{*}$. Since $i\notin -d_i$, $v_{{{-d}_{i}}}^{*}$ does not depend on $v_i'$. 
\end{proof}

Lemma \ref{change_ri} shows that given an action profile, if a buyer cannot gain any positive utility when she diffuses the information to all her neighbors, she will not gain any positive utility by reducing her diffusion effort.

\begin{lemma}\label{change_ri}
Given the buyers' action profile ${\bf a}^\prime$, if $i\notin C_w$ under $r_i' = r_i$, then $i\notin C_w$ under any $r_i' \neq r_i$.

\end{lemma}
\begin{proof}
If a buyer $i$ is not in $C_w$, then she can only be an unlucky buyer or a normal buyer.
If $i$ is an unlucky buyer, we already know that she is still an unlucky buyer no matter what $r_i'$ is. If $i$ is a normal buyer, we have $C_m \subseteq {{-d}_i}$, i.e., $C_w$ (subset of $C_m$) is independent of $r_i^\prime$. 
Therefore, a normal buyer cannot change her status by changing $r_i'$.
\end{proof}


Given the above two lemmas, we are ready to show that IDM is individually rational and incentive compatible.

\begin{theorem}\label{IDM_IC}
The information diffusion mechanism is individually rational and incentive compatible.
\end{theorem}
\begin{proof}
Assume that buyer $i$ reports her valuation truthfully. If she is a normal or an unlucky buyer, her payment is $0$ according to the payment policy. If she is an on-path buyer, her utility is $u_i(a_i,{\bf a'},(\pi^{idm},p^{idm})) = v_{{{-d}_{i+1}}}^{*}  - v_{{{-d}_{i}}}^{*}$. Since $d_{i} \supset d_{i+1}$, we have ${{{-d}_{i}}} \subset {{{-d}_{i+1}}}$ which leads to the fact that $v_{{{-d}_{i+1}}}^{*} \ge v_{{{-d}_{i}}}^{*}$. Therefore for an on-path buyer, $u_i(a_i,{\bf a'},(\pi^{idm},p^{idm}))\ge 0$. If she is the winner, then her valuation $v_i $ should be identical to the highest bid from buyers set ${- d}_{i+1}$, which is denoted as $v_{{{- d}_{i+1}}}^{*} $. Her payment is $v_{{{-d}_{i}}}^{*}$. Then her utility is $v_{{{- d}_{i+1}}}^{*} -v_{{{-d}_{i}}}^{*}$, which is non-negative. Therefore, for an arbitrary buyer, when she truthfully reports her valuation, her utility is nonnegative, IDM is individually rational.



Now we analyze the action of each kind of buyer in the social network, and prove that IDM is incentive compatible (IC). We first prove that, for all kinds of buyers, revealing the true valuation maximizes their utilities when their action $r_i'$ are determined. Based on Lemma {\ref{not_depend_valuation}}, we only need to consider what happens if a buyer $i$'s bid $v_i'$ can affect her status.

Case 1. If $i$ is an unlucky buyer,
her utility is always 0 whatever $v_i'$ she reports.

Case 2. If $i$ is a normal buyer with a truthful bid, it means that her valuation is no larger than the highest bid $v_m'$ in the network and she is outside of $C_m$, leading to $v_i\le v_m'$. As long as she bids with $v_i' \le v_m'$, $i$ cannot change $C_m$ and she will still be a normal buyer along side $C_m$, and her utility is always $0$. If $v_i' > v_m'$, she will change $C_m$. Under this circumstance, she will become the winner but now her utility decreases to $v_i-v_m'\le 0$. Therefore for normal buyers, $u_i(a_i,((v_i,r_i'),{\bf a'}_{-i}),(\pi^{idm},p^{idm})) \ge u_i(a_i,((v_i',r_i'),{\bf a'}_{-i}),(\pi^{idm},p^{idm}))$ for any $v_i'$.

Case 3. If $i$ is an on-path buyer with $v_i'=v_i$, then $i \in C_w\setminus\{w\}$. According to Eq.(\ref{IDM_allocation_rule}), we know $v_i<v_{{{-d}_{i+1}}}^{*}$. For any strategic bid $v_i'\ne v_i < v_{{{-d}_{i+1}}}^{*}$, she cannot change $C_w$ and she is still an on-path buyer. Her utility is still $v_{{{-d}_{i+1}}}^{*} - v_{{{-d}_{i}}}^{*}$. If she bids with strategic $v_i' \ge v_{{{-d}_{i+1}}}^{*}$, she will become the winner and her utility becomes $v_i - v_{{{-d}_{i}}}^{*}$. Since $v_i < v_{{{-d}_{i+1}}}^{*}$, then $u_i(a_i,((v_i,r_i'),{\bf a'}_{-i}),(\pi^{idm},p^{idm})) \ge u_i(a_i,((v_i',r_i'),{\bf a'}_{-i}),(\pi^{idm},p^{idm}))$ for any $v_i'$.

Case 4. If $i$ is the winner with truthful bid, any $v_i'\ge v_i$ will not change the allocation and her utility is $v_i - v_{{{-d}_{i}}}^{*}$. If her bid is $v_i'<v_i$, she may still be the winner with utility $v_i - v_{{{-d}_{i}}}^{*}$, or degenerate to an on-path buyer with utility ${v}_{{{-d}_{i+1}}}^{2}- v_{{{-d}_{i}}}^{*}$ where ${v}_{{{-d}_{i+1}}}^{2}$ is the second highest valuation in ${{{-d}_{i+1}}}$ and ${v}_{{{-d}_{i+1}}}^{2} < v_i$ (notice that $v_i$ is the original highest bid in $-d_{i+1}$), or she may degenerate to a normal buyer with utility $0$. In either case, her utility is not higher than truth-telling.

Therefore, we have $u_i (a_i,((v_i,r_i'), {\bf a'}_{-i}),(\pi^{idm},p^{idm})) \ge u_i (a_i,((v_i',r_i'), {\bf a'}_{-i}),(\pi^{idm},p^{idm})) $ for all $i$ and all $v_i'$.




Given buyer $i$'s declared valuation $v_i'=v_i$. We then show that for all kinds of buyers, diffusing the information to all the neighbors (i.e., $r'_i=r_i$) maximizes their utilities.

Case 1. If $i$ is a normal buyer or an unlucky buyer, according to Lemma {\ref{change_ri}}, $i$'s utility is always $0$.

Case 2. If $i$ is an on-path buyer, according to the Eq. (\ref{IDM_payment_rule}), her utility is $v_{{{-d}_{i+1}}}^{*} - v_{{{-d}_{i}}}^{*}$. When $r_i' \ne r_i$, according to Lemma {\ref{membership_di}}, $d_i$ may loss some members. Consequently, she may still be an on-path buyer. In this case, it's not hard to verify that $-d_{(i+1)'} \subseteq -d_{i+1}$ where $(i+1)'$ is the $(i+1)$th buyer in new diffusion critical sequence of the highest bid buyer. Hence, the utility of buyer $i$ may decrease since $v_{{{-d}_{(i+1)'}}}^{*} \le v_{{{-d}_{i+1}}}^{*}$. If she degenerates to the winner, then her utility is $v_i - v_{{{-d}_{i}}}^{*} < v_{{{-d}_{i+1}}}^{*} - v_{{{-d}_{i}}}^{*}$ since $v_i < v_{{{-d}_{i+1}}}^{*}$. If she degenerates to a normal buyer, her utility becomes 0. Thus for any on-path buyer $i$, diffusing the information to all her neighbors maximizes her utility.

Case 3. If $i$ is the winner, based on Lemma {\ref{membership_di}}, if $r_i' \ne r_i$, $d_{i}$ may loss some members. Since buyer $i$ is with the highest bid in $-d_{i+1}$, therefore according to allocation policy, $i$ is still the winner and her utility does not change since $-d_i$ is not affected. Therefore, $u_i (a_i,((v_i,r_i), {\bf a'}_{-i}),(\pi^{idm},p^{idm})) \ge u_i (a_i,((v_i,r_i'), {\bf a''}_{-i}),(\pi^{idm},p^{idm}))$ for all $i$ and all $r_i'$.

Putting together the above analysis, we get that $u_i (a_i,((v_i,r_i), {\bf a'}_{-i}),(\pi^{idm},p^{idm}))\ge u_i(a_i,((v_i',r_i'),{\bf a''}_{-i}),(\pi^{idm},p^{idm}))$ for all $i$, all $v_i'$ and all $r_i'$, i.e., IDM is incentive compatible.
\end{proof}

As a comparison with the VCG mechanism in social networks which may decrease the seller's revenue and even leads to a deficit, in the following theorem we prove that IDM is always weakly budget balanced (non-negative revenue for the seller). More specifically, we show that the revenue of IDM is always above that of the VCG mechanism.

\begin{theorem}
\label{them_idm_rev}
The information diffusion mechanism is weakly budget-balanced and its revenue is always greater than or equal to the revenue of the VCG in social networks.
\end{theorem}
\begin{proof}
According to the payment policy, the seller's revenue in IDM is $\sum\nolimits_{i \in C_w \setminus \{w\}} {(v_{{-d}_i}^{*} - v_{{-d}_{i+1}}^{*})} + v_{{-d}_w}^{*} = v_{{{-d}_{1}}}^{*}$, where buyer $1$ is the first buyer in $C_w$. Since $v_{{{-d}_{1}}}^{*}$ is non-negative, then IDM is weakly budget balanced.

In the VCG mechanism, if buyer $i \notin C_m$, her payment is always $0$ according to the payment policy. Therefore, the revenue of the VCG mechanism is only derived from the payments of buyers in $C_m$.
\begin{equation*}
\begin{aligned}
Rev^{VCG} &= \sum\nolimits_{i \in C_m} {(W_{{{-d}_i}}({\bf a'}_{{-d}_i}) - W_{N_{-s}\setminus\{i\}}({\bf a'}))}\\
 &= \sum\nolimits_{i \in C_m \setminus \{m\}} {(v_{{-d}_i}^{*} - {v_{N_{-s}}^{*}})} + (v_{{-d}_m}^{*} -0)\\
 &= \sum\nolimits_{i \in C_m \setminus \{1\}} {(v_{{-d}_i}^{*} - {v_{N_{-s}}^{*}})} + v_{{-d}_1}^{*}\\
 &\le  v_{{-d}_1}^{*}=Rev^{IDM}\end{aligned}
 \end{equation*}
Therefore the revenue of IDM is always equal or higher than that of the VCG mechanism.
\end{proof}

Moreover, following the proof of Theorem~\ref{them_idm_rev}, since $r_s \subseteq {{-d}_1} \cup \{1\} \subseteq {{-d}_{w+1}}$, one can easily verify that the revenue and efficiency of IDM are not less than those of the VCG mechanism without information diffusion (i.e., the second price auction involving the seller's neighbors $r_s$ only).

\section{Conclusions and Future Work}
In this paper, we generalized the classical single item auction into a social network setting, where auction information is propagated through the participants in the network and agents who cannot receive the information will not be involved in the auction. We proposed mechanisms to incentivize participants to propagate the auction information to their neighbors and to involve more buyers in the auction to increase the revenue of the seller. We found that although the VCG mechanism can be generalized to incentivize information diffusion, it will reduce the seller's revenue.
Therefore, we proposed another novel mechanism to not only incentivize information diffusion, but also simultaneously increase the social welfare and the seller's revenue.

One important future work is the study of social network based exchanges with multiple items, multiple sellers or combinatorial valuations such as ad auctions and the sharing economy platforms. Another interesting feature to consider in the current setting is the costs of the information propagation to buyers. It is also worth investigating how different social network structures affect the efficiency and revenue of our mechanism.

\section*{Acknowledgement}
This work was partially supported by the National Natural Science Foundation of China (NNSFC) under Grant No. $71601029$. Dengji Zhao also acknowledges the funding from the EPSRC-funded International Centre for Infrastructure Futures (ICIF) (EP/K012347/1).

\bibliographystyle{aaai}
\bibliography{MDS}

\end{document}